\documentclass{llncs}
\usepackage{amssymb,amsbsy,amsmath,amsfonts,amssymb,amscd}
\usepackage[utf8]{inputenc}
\usepackage{graphicx}
\usepackage{complexity}
\usepackage{tikz}
\usepackage{xspace}

\title{Efficient Generation of Stable Planar Cages for Chemistry}
\author{Dominique Barth, Olivier David, Franck Quessette, Vincent
  Reinhard, Yann Strozecki, Sandrine Vial
\thanks{Authors thank the French Labex CHARMMMAT for the
financial support of this work and David Auger for fruitful discussions about the folding algorithm. }}
\institute{Universit\'e de Versailles Saint-Quentin}

\newcommand{\alphabet} {\ensuremath{\mathcal{A}}\xspace}

\newcommand{\setofmotifs} {\ensuremath{\mathcal{M}}\xspace}
\newcommand{\Vc} {\ensuremath{V_{\text{c}}}\xspace}

\newcommand{\Next} {\ensuremath{\text{next}}\xspace}
\newcommand{\mot}[1] {\textbf{#1}\xspace}

\tikzstyle{label}=[draw=black,fill=white,text=black,circle]%
\tikzstyle{labelsat}=[draw=black,fill=lightgray,text=black,circle]%
\tikzstyle{centre}=[draw=black,fill=gray,text=white,circle]
\tikzstyle{noeud}=[draw=black,fill=white,circle]    

\begin{document}
\pagestyle{plain}
\maketitle
\begin{abstract}
In this paper we describe an algorithm which generates all colored planar maps
with a good minimum sparsity from simple motifs and rules to connect them.
An implementation of this algorithm is available and is used by chemists who want 
to quickly generate all sound molecules they can obtain by mixing some basic components.
\end{abstract}

\section{Introduction}
\label{sec:intro}

Carbon dioxide, as well as methane can be absorbed by large organic
cages~\cite{HTC10}. These cages are formed by
spontaneous assembly of small organic molecules, called motifs, bearing different
reacting centres. The prediction of the overall shape of the cage that
will be obtained by mixing the starting motifs is rather difficult,
especially because a given set of reacting partners can lead to very different cages. It is hence crucial for chemists to have an
operating tool that is capable of generating the many shapes of cages
accessible from predetermined molecular motifs. 

In this paper we present the algorithms we have designed and implemented to generates molecules that are much larger and less regular that what the chemists usually design by hand.  The molecules are modelled by maps i.e. planar embeddings of planar graphs, as explained in Sec.~\ref{sec::model}. The use of maps may seem unsuitable since they do not represent spatial positions.
Though, planar maps are a good model for spherical topologies and the embedding capture the rigidity of the motifs.
We must also be able to select the most relevant molecules among the huge number we generate.
In Sec.~\ref{sec::indices}, we characterize what a ``good'' molecule is through graph parameters
which are then used to filter the best molecules.
The relevance of our modeling and of our parameters is validated by the results we obtain: All small molecules (5-10 motifs) we generate and consider to be good according to our parameters have been studied before by chemists.
Some of the very regular molecules of medium size (10-20 motifs) we generate correspond to the largest cages chemists have ever produced. We also have produced cages of shape unknown to chemists that they now try to synthesize (see Sec.~\ref{sec::validation}).

The aim of this paper is the \textbf{generation of all colored planar maps up to isomorphism} representing possible
 molecules obtained from a set of elementary starting motifs (colors).
	As with all \emph{enumeration problems}, one difficulty is to avoid to produce a solution several times.
	Moreover the number of solutions may grow exponentially with their size, it is here the case for all bases of motifs but the most contrived.
	The complexity of such enumeration problems must then take into account the number of produced solutions (see~\cite{phd_strozecki} for more details on enumeration).
	
	We say that an algorithm is in \emph{polynomial total time} if its complexity is polynomial in the number of solutions and polynomial in the size of the produced solutions. In our context, where the number of solutions is always exponential in their size, we are interested in \emph{linear total time} algorithms.
	The best algorithms are in constant amortized time (CAT): the algorithm uses on average a constant time to generate each solution. 
	This kind of efficient algorithms exists for simple enumeration problems such as listing all trees~\cite{li1999advantages}.
	We may also want to bound the delay that is the time between the production of two consecutive solutions.
	Good algorithms have a delay polynomial, linear or even constant in the size of the generated solutions.

	There exist numerous works on enumeration and generation of planar maps~\cite{Lis85}, but none of them deals with the generation of planar maps built with a set of starting motifs and color constraints. Moreover, most of the literature deals with non-constructive tools~\cite{CV81} or yields algorithms which are not in polynomial total time.
	There are a few programs such as plantri~\cite{brinkmann2007fast} and CaGe~\cite{brinkmann2010cage} which generate efficiently some particular class of planar graphs such as cubic graphs or graphs with bounded size of face but they are not general enough for our purposes.
	
	The algorithm we present in Sec.~\ref{sec::algo} is far from being in polynomial total time since  we are not able to bound the number of isomorphic copies of each solution we generate.
	However, we will present several subroutines used in our algorithm which are either CAT, for instance the generation of paths and almost 
	foldable paths in Sec.~\ref{sec::backgen}, or in linear delay such as the folding of unsaturated maps of motifs in Sec.~\ref{sec::fold}.
	Moreover, we study several heuristics and improvements which makes the enumeration feasible for maps of medium size.
	Sec.~\ref{sec::results} presents numerical results which supports this assertion and illustrates the relative interest of our heuristics.

\section{Modeling of the problem\label{sec::model}}

In this section, we propose the modeling of our problem by maps.
A map is a connected planar graph drawn on the sphere considered up to continuous deformation. 
Note that by Steinitz's theorem, when a planar graph is 3-connected, there is only one corresponding map,
but otherwise there may be several of them. It is relevant to distinguish between two maps with the same underlying graph, since the geometrical informations 
contained in the maps are useful to the chemist who are interested in their 3D representation. 
All maps used in this paper are \emph{vertex-colored maps}.
The representation of a map is a graph and a cyclic order of the neighbors around each vertex.

We first model the basic chemical elements with maps we call \emph{motifs}. 
Then the motifs are assembled to form a \emph{map of motifs} and from this map we derive
a \emph{molecular map} that is a more faithful model of the molecular cages we try to design.

We use a finite even set of colors $\alphabet=\{a,\overline{a},b,\overline{b},c,\overline{c},\dots \}$ where 
each positive color $a$ in $\alphabet$ has a unique complementary negative color denoted
by $\overline{a}$ and $a$ is the complementary color of $\overline{a}$. Each color represents a different kind of reacting center. Let us give the definition of motifs.

\begin{definition}
A map $G=(\Vc \sqcup V,E,\Next)$ is a \emph{motif} if, 
(1)~$\Vc$ contains only one vertex $c$ called the center,
(2)~each vertex in $V$ is colored with a color in $\alphabet$,
(3)~$E = \{(c,u),~u \in V\}$, 
and (4)~$\Next$ gives an order on the edges of $c$:
$\Next((c,u)) = (c,v)$ means that the edge $(c,v)$ is "following" the edge $(c,u)$ in a clockwise drawing of $G$.
For all $k<|V|$, $\Next^k((c,u)) \neq (c,u)$ and $\Next^{|V|}((c,u))=(c,u)$.
\end{definition}
Note that a motif is a star graph. We assume as input $\setofmotifs$ a finite set of motifs all different. Each motif is
identified by a distinct color from an alphabet $\alphabet_M$ disjoint from $\alphabet$ induced by the colors existing in $\setofmotifs$.
Fig.~\ref{fig::motif} gives examples of motifs.

\vspace*{-0.5cm}
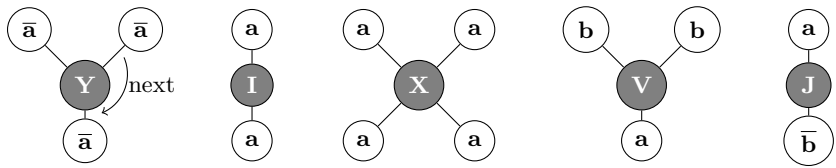
\begin{figure}
\centering  
  \begin{tikzpicture}[scale=0.74]
    \node[centre]   (C1) at (-1, 0) {$\mot{Y}$};
    \node[label] (A1) at ( 0,1) {$\mathbf{\overline{a}}$};
    \node[label] (A2) at ( -1, -1) {$\mathbf{\overline{a}}$};
    \node[label] (A3) at ( -2, 1) {$\mathbf{\overline{a}}$};
    \draw (C1) -- (A1);
    \draw (C1) -- (A2);
    \draw (C1) -- (A3);
    \draw[->] (-0.3,0.5) to[bend left=45] (-0.7,-0.5);
     \node at (0.2,0) {\Next};
    \node[centre]   (C2) at (2,0) {$\mot{I}$};
    \node[label] (B1) at (2,1) {$\mathbf{a}$};
    \node[label] (B2) at (2,-1) {$\mathbf{a}$};
    \draw (C2) -- (B1);
    \draw (C2) -- (B2);
    \node[centre]   (C3) at (5,0) {$\mot{X}$};
    \node[label] (D1) at (4,1) {$\mathbf{a}$};
    \node[label] (D2) at (6,1) {$\mathbf{a}$};
    \node[label] (D3) at (4,-1) {$\mathbf{a}$};
    \node[label] (D4) at (6,-1) {$\mathbf{a}$};
    \draw (C3) -- (D1);
    \draw (C3) -- (D2);
    \draw (C3) -- (D3);
    \draw (C3) -- (D4);
   \node[centre]   (C4) at (9, 0) {$\mot{V}$};
    \node[label] (E1) at ( 8,1) {$\mathbf{b}$};
    \node[label] (E2) at ( 10, 1) {$\mathbf{b}$};
    \node[label] (E3) at ( 9, -1) {$\mathbf{a}$};
    \draw (C4) -- (E1);
    \draw (C4) -- (E2);
    \draw (C4) -- (E3);

    \node[centre]   (C5) at (12,0) {$\mot{J}$};
    \node[label] (F1) at (12,1) {$\mathbf{a}$};
    \node[label] (F2) at (12,-1) {$\mathbf{\overline{b}}$};
    \draw (C5) -- (F1);
    \draw (C5) -- (F2);
  \end{tikzpicture}
\caption{Example of motifs on 
  $\alphabet_M=\{\mot{Y},\mot{I},\mot{X},\mot{V},\mot{J}\}$ and
  $\alphabet=\{a,\overline{a},b,\overline{b}\}$.}
\label{fig::motif}
\end{figure}

\vspace*{-0.9cm}
\begin{definition}
A connected planar map $G=(V_c \sqcup V,E,\Next)$ is a \emph{map of motifs} based on $\setofmotifs$ if,
(1)~the closed neighborhood of each vertex in $\Vc$ is a motif, 
(2)~each vertex in $V$ is connected to exactly one vertex in $V_c$ and at most one vertex in $V$. If $u$ and $v$ in $V$
are connected, the colors of $u$ and $v$ must be complementary.
The number of vertices in $\Vc$ is called the \emph{size} of $G$.
\end{definition}

Note that each motif of $\setofmotifs$ may appear any number of times in a map of motifs, it
may also be not present. A motif is a map of motifs of size 1.
In a map of motifs, a vertex of degree $1$ in $V$ is called a \emph{free vertex}.
A map of motifs with no free vertex is called \emph{saturated} otherwise it is called unsaturated.

In our implementation, we have an ordering of the edges around each element of $\Vc$ consistent with $\Next$ has been fixed.
For optimal performances, we use in our implementation a rotation map to represent a map of motif.
For each vertex $c_1\in \Vc$, it maps the $i^{th}$ edge of $c_1$, which connects $c_1$ to $u$, to a triplet
$(a,c_2,j)$ where $a$ is the color of $u$, $(c_1,u,v,c_2)$ is a path
with $c_2 \in \Vc$ and the edge $(c_2,v)$ is the $j^{th}$ of
$c_2$. The color of $v$ is necessarly $\overline{a}$
and is thus not represented. 
When $u$ is not connected to another vertex, $c_2$ and $j$ are
set to a default value.

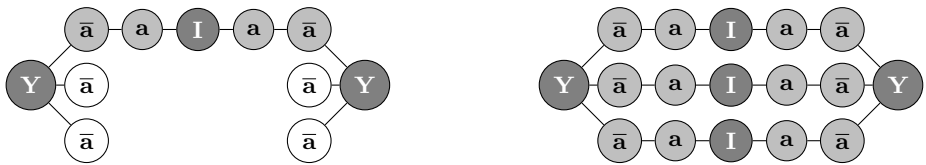
\begin{figure}[h]
  \begin{tikzpicture}[scale=0.74]
    \node[centre]   (C1) at (-1, 0) {$\mot{Y}$};
    \node[label] (A1) at ( 0,-1) {$\mathbf{\overline{a}}$};
    \node[label] (A2) at ( 0, 0) {$\mathbf{\overline{a}}$};
    \node[labelsat] (A3) at ( 0, 1) {$\mathbf{\overline{a}}$};
    \draw (C1) -- (A1);
    \draw (C1) -- (A2);
    \draw (C1) -- (A3);

    \node[centre]   (C2) at (2,1) {$\mot{I}$};
    \node[labelsat] (B1) at (1,1) {$\mathbf{a}$};
    \node[labelsat] (B2) at (3,1) {$\mathbf{a}$};
    \draw (C2) -- (B1);
    \draw (C2) -- (B2);



    \node[centre] (C5) at (5, 0) {$\mot{Y}$};
    \node[label] (F1) at (4,-1) {$\mathbf{\overline{a}}$};
    \node[label] (F2) at (4, 0) {$\mathbf{\overline{a}}$};
    \node[labelsat] (F3) at (4, 1) {$\mathbf{\overline{a}}$};
    \draw (C5) -- (F1);
    \draw (C5) -- (F2);
    \draw (C5) -- (F3);

        \draw (A3) -- (B1);
        \draw (F3) -- (B2);
  \end{tikzpicture}\hfill
  \begin{tikzpicture}[scale=0.74]
    \node[centre]   (C1) at (-1, 0) {$\mot{Y}$};
    \node[labelsat] (A1) at ( 0,-1) {$\mathbf{\overline{a}}$};
    \node[labelsat] (A2) at ( 0, 0) {$\mathbf{\overline{a}}$};
    \node[labelsat] (A3) at ( 0, 1) {$\mathbf{\overline{a}}$};
    \draw (C1) -- (A1);
    \draw (C1) -- (A2);
    \draw (C1) -- (A3);

    \node[centre]   (C2) at (2,1) {$\mot{I}$};
    \node[labelsat] (B1) at (1,1) {$\mathbf{a}$};
    \node[labelsat] (B2) at (3,1) {$\mathbf{a}$};
    \draw (C2) -- (B1);
    \draw (C2) -- (B2);

    \node[centre]   (C3) at (2,0) {$\mot{I}$};
    \node[labelsat] (D1) at (1,0) {$\mathbf{a}$};
    \node[labelsat] (D2) at (3,0) {$\mathbf{a}$};
    \draw (C3) -- (D1);
    \draw (C3) -- (D2);

    \node[centre]   (C4) at (2,-1) {$\mot{I}$};
    \node[labelsat] (E1) at (1,-1) {$\mathbf{a}$};
    \node[labelsat] (E2) at (3,-1) {$\mathbf{a}$};
    \draw (C4) -- (E1);
    \draw (C4) -- (E2);

    \node[centre]   (C5) at (5, 0) {$\mot{Y}$};
    \node[labelsat] (F1) at (4,-1) {$\mathbf{\overline{a}}$};
    \node[labelsat] (F2) at (4, 0) {$\mathbf{\overline{a}}$};
    \node[labelsat] (F3) at (4, 1) {$\mathbf{\overline{a}}$};
    \draw (C5) -- (F1);
    \draw (C5) -- (F2);
    \draw (C5) -- (F3);

        \draw (A3) -- (B1);
        \draw (A2) -- (D1);
        \draw (A1) -- (E1);
        \draw (F3) -- (B2);
        \draw (F2) -- (D2);
        \draw (F1) -- (E2);
  \end{tikzpicture}
\caption{Example of two maps of motifs based on $\setofmotifs
  = \{\mot{Y},\mot{I}\}$, the first map is unsaturated while the second
  map is saturated.}
\label{fig::mapofmotifs}
\end{figure}

Based on a saturated map of motifs we construct the molecular map that is the graph model of the cages.
\begin{definition}
Let $G=(\Vc \sqcup V,E_G,\Next_G)$ be a saturated map of motifs based on $\setofmotifs$, we define the \emph{molecular map}
$M$ as the map $G$ where all paths of size three between vertices of $V_c$ are replaced by an edge.
\end{definition}
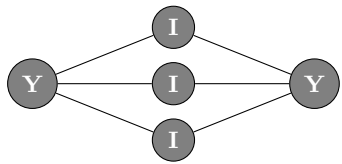
\begin{figure}[h]
 \centering
 \begin{tikzpicture}[scale=0.75]
   \node[centre]   (C1) at (0, 0) {$\mot{Y}$};
   \node[centre]   (C2) at (2.5,1) {$\mot{I}$};
   \node[centre]   (C3) at (2.5,0) {$\mot{I}$};
   \node[centre]   (C4) at (2.5,-1) {$\mot{I}$};
    \node[centre]   (C5) at (5, 0) {$\mot{Y}$};
    \draw (C1) -- (C2);
    \draw (C1) -- (C3);
    \draw (C1) -- (C4);
    \draw (C5) -- (C2);
    \draw (C5) -- (C4);
    \draw (C5) -- (C3);
  \end{tikzpicture}

 \caption{The molecular map corresponding to the saturated map of
   motifs in Fig. \ref{fig::mapofmotifs}}
 \label{fig:molmap}
\end{figure}

\section{Description of the algorithm\label{sec::algo}}

	  The aim of this paper is to solve the following problem: 
	  given 
	  a base of motifs $\setofmotifs$ and an
	  integer $n$, enumerate all molecular maps of size $n$ based on $\setofmotifs$. The complexity depends only on $n$ since 	
	  the size of $\setofmotifs$ and the size of its elements are assumed to be small constants (usually less than 4).
	  In this section, we describe an algorithm which solves this problem and explain in details its two main steps.
	  	  
	  The first one, \emph{the concatenation}, consists in adding
	  edges between complementary vertices of two maps of motifs in such a way the result is still a map of motifs. 
	  In this paper, we always concatenate a single motif to a map of motifs, see~\cite{barth2013map} for other concatenations. 
	  Sec.~\ref{sec::backgen} presents the different strategies of concatenation.	 
	  The second, \emph{the fold or folding}, consists in adding an edge between two complementary vertices of a map of motifs,
	  in such a way the result is a map of motifs. Sec.~\ref{sec::fold} presents an efficient approach to folding that we use to saturate the maps obtained by concatenation. 
	  Then, Sec.~\ref{sec::isomorphism} explain how we detect and discard isomorphic copies of the same graph. Finally in Sec.~\ref{sec::indices}, we introduce the indices which characterize a good molecular map and explain how we compute them.

	  \subsection{Backbone generation\label{sec::backgen}}

		The first step is to generate all \emph{backbones}, that is unsaturated maps of motifs of a given size $n$ which are of a very simple shape.
		The aim is that, by folding these backbones in a second step, we will recover all saturated maps of motifs.
		Since every map of motifs have a spanning tree, we can choose trees as backbones and be sure to recover all saturated maps.
		But for performance reason, we will also use paths and cycles as backbones. This turns out to be good \emph{heuristics}, speeding up considerably our algorithm while only mildly reducing the set of generated maps of motifs.
		We would also like to restrict the backbones to those which can be folded into some saturated map. 
		We address this problem by enumerating only what we call \emph{the almost foldable backbones}, with a complexity as good as for the generation of regular backbones.
		This new algorithm greatly improve the computation time.
		
		\paragraph{Spanning tree.}
		
		In a first version of our algorithm~\cite{barth2013map},
		the set of non isomorphic trees of size $n$ was explicitly stored.
		To produce the set of trees of size $n+1$, a single motif of every possible color
		was concatenated to each free vertex of each tree of size $n$. This generates all trees of size $n+1$, but 
		the drawback is that some trees are generated several times. The algorithm was thus not in linear total time 
		and we needed to do an isomorphism test on every generated tree. 
		We now generate all trees where the root and its first edge are fixed with a simple CAT algorithm. 
		This method generates a tree as many times as edges in the tree: one for each choice of a vertex as root and for each choice of first edge of this root. Therefore, the implemented algorithm do not need to store the trees which are produced on the fly,
		and has a linear delay.
		A way to further improve this would be to use ideas from CAT algorithms which generate unrooted trees~\cite{li1999advantages}. The main idea is to choose as root the centroid of the tree. However we have to deal with a second and harder problem: we generate maps of motifs and their vertices are colored.	
		We can generate all maps of motifs sharing the same underlying tree efficiently but they may turn out to be isomorphic.

		\paragraph{Hamiltonian paths.}

		Since generating trees is not easy, we propose to use simpler objects as backbones, here maps of motifs such that
		all vertices of $\Vc$ are on a path. These maps are caterpillar trees, but since the elements of $\Vc$ on the central path entirely determine the elements at distance one, we will consider them as paths and call them so.
		There are two advantages to generating paths instead of trees: they are easier to generate and their number is smaller. 
		The drawback is that not any planar graph has an Hamiltonian path, therefore we could miss some 
		planar maps in our enumeration. However, most small planar graphs have an Hamiltonian path, 
		for instance all planar cubic 3-connected graphs of size less than 38~\cite{holton1988smallest} and, if Barnette's conjecture holds, all fullerene graphs.
		
		The regularity of the graphs (all vertices of the same degree) crucially matters in the existence of an Hamiltonian path. 
		Consider for instance the base of motifs $\setofmotifs = \{\mot{I},\mot{Y}\}$ from Fig.~\ref{fig::motif}.
		All molecular maps based on $\setofmotifs$ are bipartite graphs: the \mot{I}'s in one set of the bipartition and the \mot{Y}'s 
		in the other. But in saturated maps of motifs, we have twice the number of \mot{Y} equal three times the number of \mot{I} because all vertices in $V$ must be connected, therefore there are no Hamiltonian path except in graphs with exactly three \mot{I} and two \mot{Y}. 
		This problem can be easily solved by building from $\setofmotifs$ a new base of motifs which in the end generates the same molecular maps (see Sec.~\ref{sec::metamotifs}).
		
		Let us now explain how we generate all paths based on a set of motifs $\setofmotifs$.
		We first build for each letter $a \in \alphabet$ a list $L_a$ of all non isomorphic motifs whose first 
		edge is incident to a vertex of label $\bar{a}$.
		This data structure allows us to have a complexity independent of the size of $\setofmotifs$ and of $\alphabet$. 
		Then to build all possible paths of size $n+1$ from a path of size $n$, we consider its last vertex $c\in \Vc$ and for each of the free vertex $v$ connected to $c$ and of color $a$, we attach every motif of $L_a$. Remark that beginning by the empty path, we generate all possible paths of a given size by applying recursively the algorithm.
		If we consider the paths as rooted at the first vertex produced during the algorithm, every path generated is clearly different.  
		However, we can also consider the last concatenated vertex as the beginning of the path, which means  
		we generate every path but the palindromes twice. To avoid that, we put an ordering on $\alphabet_M$, the colors of the center vertices, and we consider the sequence of colors in a path. If the sequence of colors from the beginning to the end is lexicographically larger than the sequence from the end to beginning we output the path otherwise we do not. This is implemented in our algorithm and adds only in average a constant time. 
		
		\begin{proposition}
		 The previous algorithm produces all maps of motifs which are paths without redundancies in constant amortized time, when in the base of motifs no two motifs of degree $2$ can be concatenated.
		\end{proposition}
		\begin{proof}
		 The tree of recursive calls of our algorithm can always be seen as of degree at least $3$ by merging nodes of degree $2$ to nodes of degree larger.
		 Therefore it has at least as many internal nodes as leaves which correspond to output solutions. Since the algorithm needs only a constant time to go from one node to another, the generation of all paths can be done in constant amortized time. 
		 \qed
		\end{proof}
		  
		 In our practical examples, there are never motifs of degree two which can be concatenated. Without this condition, the algorithm has still a linear delay. 
			
		\paragraph{Hamiltonian cycles.}
		
		If we want to further restrict the backbones we generate, a simple idea is to consider cycles instead of
		paths. Again it is a good choice if all motifs have the same degree or can be made so, since for instance all planar cubic 3-connected graphs of size less than 23 have an Hamiltonian cycle~\cite{aldred1999cycles}.
		Moreover, we will only generate $2$-connected graphs and not the ones which are only $1$-connected.
		It is a desirable side effect, since those graphs have a bridge they are always the worse for the two main indices
		we are interested with, i.e. the minimum sparsity and the size of the largest cycle (see Sec.~\ref{sec::indices}).
		
In our implementation, we obtain the cycles by generating every path and by connecting their beginning to their end when possible. The same cycle can be obtained from several different paths (at most as much as its number of vertices). Therefore our algorithm is in linear amortized time. The question is, can we generate all cycles with a CAT algorithm? It seems hard because we cannot fix a natural first vertex in a cycle as in a tree, since all its vertices can be isomorphic.

		\paragraph{Almost foldable paths.}  

		In each backbone we build, all free vertices will eventually be folded to get a saturated map of motifs.
		A simple necessary condition on the colors of a saturated map of motifs	is that for each color $a \in \alphabet$, there are as many vertices in $V$ labeled by $a$ and $\bar{a}$. 
		A backbone which satisfies this condition is said to be \emph{almost foldable}.
		Let $G$ be a map of motifs and let $a_1,\dots,a_k$ be the positive colors of the alphabet $\alphabet$.   
		We denote by $C_G$ the characteristic vector of $G$, it is of size $k$ and its $i^{th}$ component is the number 
		of elements in $V$ labeled by $a_i$ minus the number of elements labeled by $\bar{a}_i$.
		Note that a map $G$ is almost foldable if and only if $C_G$ is the zero vector.

		We propose here a method to generate in constant amortized time only the almost foldable paths.
		We introduce a function $F : \mathbb{N}\times \alphabet \times \mathbb{Z}^k \rightarrow  2 ^ \alphabet$ 
		which has the following semantic: $a' \in F(n,a,(c_1,\dots,c_k))$ if and only if (1)~there is a path $P$ of size $n$
		with a free vertex in the first motif labeled by $a$, (2)~$C_P = (c_1,\dots,c_k)$, (3)~a vertex of the last motif is labeled by $a'$.

		\begin{proposition}
		 There is an algorithm which enumerates all almost foldable paths in constant amortized time plus a precomputation in $O(n^{k+1})$, when in the base of motifs no two motifs of degree $2$ can be concatenated.
		\end{proposition}
		
		\begin{proof}
		First, we explain how to generate all needed values of the function $F$ in time $O(n^{k+1})$ by dynamic programming. 
		Denote by $f$ the maximal number of vertices in a motif labeled by the same color. 
		For a path $P$ of size $n$, it is clear that the coefficients in $C_P$ are all in the interval $[-nf, nf]$. 
		Therefore, to generate paths of size $n$, since $f$ and the size of $\alphabet$ are constants, we need to store $O(n^{k+1})$ values of $F$ only. 
		
		$F$ is easy to compute for $n=1$: we consider each motif $M \in \setofmotifs$ 
		and each $v$ of label $a$ in $M$, and let $F(1,a,C_M)$  be the set of labels of all vertices of $M$ but $v$.
		Assume we have generated the values of $F$ for $n$, we generate the values for $n+1$ in the following way. 
		For each $a$, $C$ and each  $a' \in F(n,a,C)$, we consider all motifs  $M \in \setofmotifs$ such that one of their vertex is labeled by $\bar{a}$. 
		We add all the labels of the other vertices to the set $F(n,a,C+C_M)$. This algorithm only does a constant number of operations 
		for each value of $F$ it computes, therefore its complexity is $O(n^{k+1})$.
		
		Now that $F$ is computed, we use it in our path generation algorithm to generate only the almost foldable paths. 
		Assume we have generated a path $P$ of size $n'$, its characteristic vector $C_P$ and we want to add a node at the end by connecting it to a node of label $a$. 
		Assume we have already computed $C_P$.	The algorithm checks if $F(n-n',\bar{a},-C_P)\neq \emptyset$.
		If it is the case the algorithm go on normally otherwise it backtracks since this extension cannot yield a non foldable path.
		This improvement only adds a single test at each step of the original algorithm, plus an addition of a constant sized vector to maintain the value of $C_P$. Therefore it is in constant amortized time.
		\qed
		\end{proof}

			The complexity of the precomputation may seem to be large but $k$ must be seen as a small constant (less than $4$).
			It is negligible with respect to the generation of paths, which is exponential in $n$ because of the number of non isomorphic paths. In practice, the precomputation takes only a few milliseconds for size of graphs up to $40$ on a regular desktop computer. On the other hand, this optimization makes the time to computes all the backbones much smaller than 
			the time to do the next steps.

			\paragraph{Almost foldable trees.} 
			
		Following the idea used to efficiently compute almost foldable paths, we give here two ways to generate
the almost foldable trees.

 When we extend a tree by a concatenation, it can be through any vertex. To keep the same 
 dynamic programming algorithm as for paths we should track all free vertices in the tree in construction, which would make the algorithm exponential time. 
 There are two solutions to this problem, the first and the one we have implemented is to compute a multidimensional array $A$ such that $A(n,C) = 1$ if there is a \emph{forest} $F$ of size $n$ such that $C_F = C$ and  $A(n,C) = 0$ otherwise.
We can thus test in our algorithm generating trees, whether any partial tree can be extended to a structure of the right size by a forest. 
Since we generate trees and not forests, we will sometimes expand a partial tree and obtain no almost foldable backbone in the end.
		
		The second solution is to change the characteristic vector of a backbone so that each of its component is the number of free vertices of some color
		positive or negative. In this way it is easy to compute an array $A$ such that $A(n,C) = 1$ if there is a \emph{tree} $T$ of size $n$ such that $C_T = C$ and  $A(n,C) = 0$ otherwise. Indeed, for each motif $M$ with a free vertex of color $a$, if for some $C$ $A(n,C) = 1$ and $C$ has a non-zero component $\bar{a}$ then there is a tree 
		of size $n+1$ with vector $C +C_M$ that is $A(n,C+C_M)=1$. The only drawback is that the size of $A$ and thus the complexity of the precomputation is $O(n^{2k+1})$, where $k$ is the number of positive colors while the size of $A$ in the solution we have implemented is $O(n^{k+1})$.

	\subsection{Folding of the backbones}\label{sec::fold}
				
		Let $G$ be a map of motifs, the \emph{fold} operation on the vertices $u$ and $v$ is adding the edge $(u,v)$ to $G$.
		The operation is valid if $u$ and $v$ are free, of complementary colors and in the same face of $G$. Therefore,
		the graph obtained after the fold is still a map of motifs. 
		In this section we generate from a backbone, by sequences of folds, all possible saturated maps of motifs. 
		
		The \emph{outline} of a face is the list in order of traversal of the free vertices.
		An outline is a circular sequence of vertices $(v_1,\dots,v_n) \in V^n$. 
		Sequence means that the order is significant and circular means that the starting point is not. For instance,
		$(v_1,v_2,v_3)$ and $(v_3,v_1,v_2)$ are the same circular sequence but are different from $(v_3,v_2,v_1)$.
		Remark that a tree or a path has a single outline, a cycle has two and a saturated map has only empty outlines.
		The color of an outline $(v_1,\dots,v_n)$ is the word $w_1\dots w_n$ with $w_i$ the color of $v_i$.
		Folding two vertices $v_i$ and $v_j$ in the same outline of color $W_1w_iW_2w_jW_3$ creates two outlines of color $W_3W_1$ and $W_2$. 
		The fold operation can then be seen as an operation from words over $\alphabet$ to multiset of words.
		Remark that this operation is very similar to the reduction of consecutive complementary parentheses which enables to define
		the classical Dyck language of balanced string parentheses. 
		
                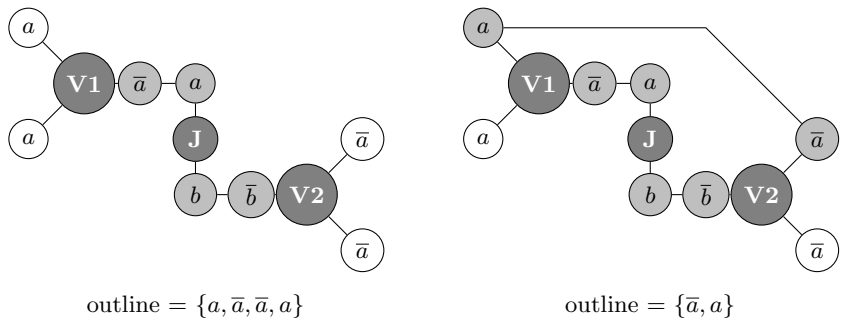
\begin{figure}[h]
                  \centering
                  \hfill
                    \begin{tikzpicture}[scale=0.74]
                      \node[centre] (CJ) at (0,0) {\mot{J}};
                      \node[labelsat] (J1) at (0,1) {$a$};                 
                      \node[labelsat] (J2) at (0,-1) {$b$};
                      \node[centre](CV1) at (-2,1) {$\mot{V1}$};
                      \node[labelsat] (V11) at (-1,1) {$\overline{a}$};
                      \node[label] (V12) at (-3,2) {$a$};
                      \node[label] (V13) at (-3,0) {$a$};
                      \node[centre](CV2) at (2,-1) {$\mot{V2}$};
                      \node[labelsat](V21) at (1,-1){$\overline{b}$};
                      \node[label](V22) at (3,0) {$\overline{a}$};
                      \node[label](V23) at (3,-2){$\overline{a}$};
                      \draw (CJ) -- (J1);
                      \draw (CJ) -- (J2);
                      \draw (J1) -- (V11);
                      \draw (J2) -- (V21);
                      \draw (CV1) -- (V11);
                      \draw (CV1) -- (V12);
                      \draw (CV1) -- (V13);
                      \draw (CV2) -- (V21);
                      \draw (CV2) -- (V22);
                      \draw (CV2) -- (V23);
                      \node at (0,-3) {outline = $\{a,\overline{a},\overline{a},a\}$};
                    \end{tikzpicture}\hfill
                    \begin{tikzpicture}[scale=0.74]
                      \node[centre] (CJ) at (0,0) {\mot{J}};
                      \node[labelsat] (J1) at (0,1) {$a$};                 
                      \node[labelsat] (J2) at (0,-1) {$b$};
                      \node[centre](CV1) at (-2,1) {$\mot{V1}$};
                      \node[labelsat] (V11) at (-1,1) {$\overline{a}$};
                      \node[labelsat] (V12) at (-3,2) {$a$};
                      \node[label] (V13) at (-3,0) {$a$};
                      \node[centre](CV2) at (2,-1) {$\mot{V2}$};
                      \node[labelsat](V21) at (1,-1){$\overline{b}$};
                      \node[labelsat](V22) at (3,0) {$\overline{a}$};
                      \node[label](V23) at (3,-2){$\overline{a}$};
                      \draw (CJ) -- (J1);
                      \draw (CJ) -- (J2);
                      \draw (J1) -- (V11);
                      \draw (J2) -- (V21);
                      \draw (CV1) -- (V11);
                      \draw (CV1) -- (V12);
                      \draw (CV1) -- (V13);
                      \draw (CV2) -- (V21);
                      \draw (CV2) -- (V22);
                      \draw (CV2) -- (V23);
                      \draw (V12) -- (1,2)--(V22);
                      \node at (0,-3) {outline = $\{\overline{a},a\}$};
                    \end{tikzpicture}
                    \hfill
                  \caption{A map on
                    $\alphabet_M=\{\mot{V1},\mot{V2},\mot{J}\}$ and its outline before
                    and after a fold operation.}
                  \label{fig:outlineandfold}
                \end{figure}
                
		Applying a sequence of fold to a backbone to get a saturated map is the same as applying a sequence of reductions to the colors of an outline so that we obtain only empty words.
		We work from now on only on the words $w_1\dots w_n$ and on sequences of reductions.
		If in a sequence of reductions, the reduction is applied to $w_i$ and $w_j$ we say that the sequence \emph{pairs} $i$ with $j$.
		
		Let us call a word (or a multiset of words) which reduces to a multiset of empty words a \emph{foldable word}.
		As in the case of parentheses languages, we can restrict the reduction to consecutive complementary letters which transforms $W_1a\overline{a}W_2$ into the word $W_1W_2$. 
		Indeed, when a word is foldable, it can be reduced to empty words using \emph{the restricted reduction} of consecutive letters only by reordering the sequence of reductions.		
		We call \emph{result} of a sequence of reductions the set of pairs $(i,j)$ such that the sequence has paired $i$ and $j$. The previous remark shows that it is indeed a set of pairs
		and not a sequence.
                Our aim is to generate all different results of sequences of reductions on foldable words without redundancies.

		\begin{lemma}[Folklore]\label{lemma:folklore}
		 The restricted reduction on words is confluent i.e. each sequence of restricted reduction starting from a foldable word can be extended so that we get an empty word. 
		\end{lemma}
	  \begin{proof}
		 To prove our lemma, it is enough to prove that if $S$ is the sequence which reduces a 
		 word $W= W_1w_iw_{i+1}W_2$ with $w_i=\overline{w_{i+1}}$, then $W_1W_2$ is foldable.
		 If $S$ pairs $i$ and $i+1$, then $W_1W_2$ can be reduced to the empty word by $S$.
		 We now assume that $S$ pairs $w_i$ with $w_k$ and $w_{i+1}$ with $w_l$, where $W = W_1^1w_kW_1^2w_iw_{i+1}W_2^1w_lW_2^2 $.
		 Remark that the case where $S$ pairs $w_i$ with $w_l$ and $w_{i+1}$ with $w_k$ is not possible because 
		 all letters between $i$ and $l$ must be paired together by definition and $i+1$ is between $i$ and $l$ but not $k$.
		 Inside the sequence $S$, we can find subsequences which reduce
		 $W_1^2$, $W_2^1$ and $_W1^1W_2^2$ to empty words since we are allowed to reduce consecutive letters only. Therefore $W_1W_2 = W_1^1w_kW_1^2W_2^1w_lW_2^2$ can be reduced to the empty word. First the sequences reducing $W_1^2W_2^1$ are used to obtain the word $W_1^1w_kw_lW_2^2$. Then one step of reduction remove $w_kw_l$ which are of complementary color by definition. Finally we obtain 
		 $W_1^1W_2^2$ which is foldable.	\qed	  \end{proof}

	    As a consequence of this lemma, we get a simple algorithm for testing whether a word is foldable:
	    reduce the word as long as it is possible and if an empty word is obtained, the word is foldable.
	    \begin{proposition}    
 There is a linear time algorithm to test whether a word is foldable.
\end{proposition}
\begin{proof}
 The word is represented by a doubly linked list of its letters. At a given step of the algorithm we are at some position $i$ in the list.
 If the letters at position $i$ and $i+1$ in the list are complementary, they are removed and $i$ is set to be $i-1$ if possible, $0$ otherwise.
 If the letters are not complementary, $i$ is incremented. The algorithm stops and decides that the word is foldable when the list is empty.
 If $i$ is at some point the last element of the list then the algorithm stops and decides that the word is not foldable.
 The algorithm is clearly in linear time, since at each step either the size of the list decreases or the current position increases.
 Finally this algorithm is correct, because if it stops without removing every element in the list, it means that there are no two consecutive
 complementary letters left. Therefore there are no possible further restricted reductions and the obtained word is not foldable. By Lemma \ref{lemma:folklore}, since the reduction is confluent, 
 the original word is also not foldable. \qed	
\end{proof}
	    
	    We use this algorithm each time we produce a backbone to test whether it can be folded into a saturated map of motifs.
            Note that, even if we generate almost foldable backbones only, we may generate some which are not foldable such as those with outline $ba\bar{b}\bar{a}$.

            \begin{proposition}
             There is an algorithm which enumerates all distinct results of sequences of reduction on a foldable word,
             with a linear delay and a quadratic precomputation.
            \end{proposition}
            \begin{proof}
            
            For a given word $W$ we first build the lists $L_i$ which contain the set of indices $j > i$ such that 
            $w_i$ can be folded with $w_j$ and the obtained set of words is still foldable.
            
            The lists $L_i$ are built from a boolean matrix $M$ such that $M_{i,j}$ is true if and only if 
             the word $w_i\dots w_j$ is foldable. The matrix is computed by dynamic programming:
            $M_{i,i+1}$ is true if and only if $w_i$ and $w_{i+1}$ are complementary. We compute $M_{i,j}$ once we have computed 
            all $M_{i',j'}$ such that $(j'-i') < (j-i)$ by using the fact that $w_i\dots w_j$ is foldable if and only if $w_i\dots w_k$
            and $w_{k+1} \dots w_j$ are foldable for some $k$ in $[i+1,j]$ or $w_i$ and $w_j$ are complementary and $w_{i+1}\dots w_{j-1}$ is foldable. By this method, the matrix $M$ is computed in time cubic in the size of the word. 
            In fact, by Lemma \ref{lemma:folklore}, if there is a $k$ such that $w_i\dots w_k$
            and $w_{k+1} \dots w_j$ are foldable, then for all $l$ such that $w_i\dots w_l$ is foldable, then  $w_{l+1} \dots w_j$ is foldable.
            We store for each $i$ the smallest $k>i$ such that $w_i\dots w_k$ is foldable. Hence we can decide whether there is a $k$ such that $w_i\dots w_k$ is foldable in constant time and we compute the matrix $M$ in quadratic time.

            Remark that a sequence of reductions applied to a word $W$ yields a set of subwords which are consecutive letters of $W$. Therefore we can represent the result of several reductions by a set of pairs 
            $\{(l_1,r_1),\dots,(l_k,r_k)\}$ with $(l_i,r_i)$ representing the word $w_{l_i}\dots w_{r_i}$ and $l_i < r_i < l_{i+1}$.
            We  build the results of sequences of reductions in a recursive way.        
            Assume we have already built a result $R$ through a sequence of reductions applied to $W$, which has produced the set $\{(l_1,r_1),\dots,(l_k,r_k)\}$.
            We consider $l_1$, the index of the first letter which has not been reduced and we do the reduction with every possible letter of index 
            $i \in [l_1,r_1] \cap L_{l_1}$ which produces the set $\{(l_1,i),(i+1,r_1)\dots,(l_k,r_k)\}$ and the result $R\cup \{(l_1,i)\}$.
            By using recursively this algorithm starting on $W$, we obtain all possible results $R$ corresponding to a reduction to a multiset of empty words.
            It is not possible to generate twice a result since at any point of the algorithm we make recursive calls on $R\cup \{(l_1,i)\}$
            for different values of $i$ which makes the results produced by each call disjoint.
            Between two recursive calls we do only a constant number of operations, therefore the delay 
            is bounded by the depth of the tree of recursive calls, that is the size of the word $W$.   \qed 
            \end{proof}
          
	    The enumeration algorithm we have described is \textbf{exponentially better} than the naive one where each possible letter is 
	    folded when it is next to a complementary letter and so on recursively. The complexity of the naive algorithm is proportional to the number of sequences of reductions while our is proportional to the number of results. For instance, on words of the form $W^n$ with 
	    $W = a\bar{a}\bar{a}a$, there is only one result but $(2n)!$ sequences of reductions.
%

	\subsection{Dealing with isomorphic copies}\label{sec::isomorphism}
	
	 Since the construction process does not guaranty uniqueness of the generated maps, 
	we need to detect during the enumeration the isomorphic copies of already generated maps to discard them.
	To do that we need to compute a unique signature for each map and 
	  we must store all produced maps and their signatures. 
	  Since the number of maps grows exponentially with their size, they are stored in a dynamic set structure 
	  which supports logarithmic addition and research of elements. In our implementation we have used an AVL whose key is 
	  the signature. Hence each time a new map is produced, we compute its signature and if this signature is already in the AVL, it is simply not inserted.

	From a theoretical point of view, planar isomorphism is well understood since it has been proved to be solvable 
	in almost linear time~\cite{hopcroft1974linear} and logarithmic space~\cite{datta2009planar}. 
	However this algorithm is not practical and hard to implement as observed in~\cite {kukluk2003algorithm},
	especially if we want a signature rather than just an isomorphism test.
	This is particularly true for our small graphs of size about $20$, which is the reason why we rely on a 
	simpler algorithm of quadratic complexity in the spirit of \cite{weinberg1966simple}.
	 The idea is that in a map,
	when a first edge is fixed we can do a \emph{deterministic} traversal of the graph using the order on each neighborhood.
	The signature is the least lexicographic traversal amongst the traversals beginning by all edges of the map.

	Let us describe precisely the quadratic isomorphism algorithm. 
	All the signatures are numbers in a base $B$
	with $B=n+|\alphabet|+|\alphabet_M|$.
	The first step that is common to all the maps of motifs of the same size $n$ is to assign to each color 
	in $\alphabet$ and $\alphabet_M$ a different digit in $[n , n+|\alphabet|+|\alphabet_M|)$ in base $B$.
	In a map of motifs $G=(V_c,V,E,\Next)$ of size $n$ and for any edge 
	$(c,u) \in E$ with $c\in \Vc$ we perform a deterministic depth first search that will define the signature of $G$ starting at $(c,u)$.
	Since signatures are numbers, they can be easily compared and \emph{the signature} of $G$ will be the minimum number
	over all starting points.
	
	For computing a signature starting at $(c,u)$, at first visit of each vertex in $V_c$ assign an index number that is a digit in the range $[0,n)$ in the base $B$.
	From $c$ visit its neighbor $u$: since the map is saturated $u$ is connected to a vertex $v \in V$
	and $v$ is connected to a vertex $c' \in \Vc$.
	Construct the signature by concatenating
	the index number of $c$, the digits of the colors $c$, $u$, $v$ and $c'$ and the index number of $c'$.
	If $c'$ is already visited backtrack and continue the visit from $(c,next((c,u)))$ else continue the visit starting at $(c',v')$ with $(c',v')=\Next((c',v))$ 
	and so on until all quadruplets $(c,u,v,c')$ are visited once.
	At the end, we obtain a signature in base $B$ for the starting point $(c,u)$.
	Note that the signature itself is of size linear in $n$. Given any signature one may exactly reconstruct the graph. 
	Conversely two graphs which are isomorphic have the same signature because the signature computation does not take into account the order or name of the nodes.

	We make a simple optimization, which is crucial, since profiling our algorithm reveals that it spends more than half of its time computing signatures.
	We assign the lower digits to the colors of $c$ and $u$ such that the number of couples $(c,u)$ is minimal and non zero. Since the signatures are constructed with the most significant bit first, during the construction of a signature, we test for each digit added if the signature is at this point greater than the minimal one. Thus we can cut very efficiently in the signature calculation process.
	
Moreover, the computed signature allows to detect chiral molecules, a very important notion in chemistry. Two maps are chiral if
one is isomorphic to the other when the order of the next predicate is reversed for all neighborhoods.

\subsection{Indices computed on the molecular map}\label{sec::indices}
	A molecular map is a candidate to be a ``good'' cage for chemistry.
	The definition of a ``good'' cage is merely topological: the 3D shape must be close to a sphere, 
	it must be resistant to deformations and cuts and it must have an "entrance".
	We are able to check if a molecule satisfies or not these requirements only by considering the structure of its molecular map:
	First the map is \emph{planar} and \emph{connected} by construction. In quadratic time we compute the \emph{equivalence classes of vertices up to automorphism}, using the same technique as to compute a signature, which helps measure the sphericity of the cage. 
	The entrance is given by the \emph{size of its largest face}, which is easily computed in linear time.
	The resistance of a map is given by its \emph{minimum sparsity}. 
	
From a large set of experiments, these indices have proved to be realistic to the chemist on several examples (see Sec.~\ref{sec::validation}). They are then used in our implementation to limit the number of molecular maps output by the program, which would otherwise be in such great number that a chemist could not try to study them all. For instance, all maps with a small minimum sparsity are filtered out.

	\paragraph{Distribution of the sizes of faces}\label{sec::distfaces}
		The faces size is an important parameter in the cage construction. The chemist wants a cage with an "entrance". In graph terms we seek for graphs with one large face and all the others faces of size around the mean size, which makes the molecule more spherical in practice.
		The distribution of the face size is straightforward to compute. As an indicator we compute the size difference between the two largest faces divided by the mean size.
		This indicator is zero when there is two largest faces with the same size and grows with the entrance size.
		
	\paragraph{Equivalence classes of the vertices}\label{sec::auto}
		Two vertices (motifs) of a molecular map are in the same class if it exists an automorphism that send one to the other.
		We compute the equivalence classes of all vertices: If the signature starting form $(c_1,u_1)$ is equal to the signature starting at $(c_2,u_2)$ the motif centered on $c_1$ is in the same class as he motif centered on $c_2$.		
		The chemist, when synthesizing a molecule corresponding to a molecular map,
		will use the same compound for all motifs in the same equivalence class. In addition the less the number of classes the more the molecule has a spherical shape. 
	
	\paragraph{Minimum sparsity}
	We now define the sparsity and explain how to compute it, since it is 
	the most relevant index and the hardest to compute. 
	A cut of a graph $G=(V,E)$ is a bipartition of $V$.
	The \emph{size} of a cut $S=(S_1,S_2)$ is the number of edges with one end in $S_1$ and the other in $S_2$.
	The sparsity of a cut is $sparsity(S) = \frac{size(S)}{\min(|S_1|,|S_2|)}$.
	The Sparsest Cut problem is to find the minimum sparsity over all cuts.
		We first implemented a brute-force algorithm, using a Gray code which enumerates all possible partitions of the set of vertices
		in time $O(2^n)$ where $n$ is the number of vertices in our graph. Since we were using a Gray code, the partition 
		changes at each step by only one element and the cut can be computed in constant time from the previous one.
		Therefore we have a simple algorithm with complexity $O(2^n)$ where $n$ is the number of vertices in our graph,
		which is useful for $n$ up to twenty but not practical for larger sizes.
				
		Although computing the minimum sparsity is $\NP$-complete in general (minimum cut into bounded set in~\cite{garey1979computers}),  there is a polynomial time algorithm when the graph is planar~\cite{park1993finding}. Since the time to compute the minimum sparsity was the limiting factor of our program, we have implemented and adapted to our case this more complicated algorithm (which has never been done as far as we know).
		
		The main idea is that a cut in a graph corresponds exactly to a cycle in the dual graph (see \cite{diestel2005graph} for graph definitions useful in this paragraph).
		A weight is associated to each cycle of the dual: if the corresponding cut in the primal partitions it into $S_1$ and $S_2$, the weight is $\min(|S_1|,|S_2|)$.
		From a spanning tree of the dual, we build a base of its fundamental cycles. A fundamental cycle is given by any edge not in the spanning tree completed by edges of the spanning tree to form a minimal cycle. 
		From symmetric differences of fundamental cycles, we can generate every cycle and its weight.
		
		For each edge in the dual, we build a graph such that paths from a given vertex correspond to cycles of the dual which use the edge.
		Moreover, the weight of the cycle can be read in the last vertex of the path, and the size of the corresponding cut is the length of the path.
		Therefore, computing a single source shortest-path in each of these graphs enables us to compute the value of the sparsest-cut. While in the original article this was done 
		by a modified Dijkstra algorithm, we use a breadth first-search. This is faster and it enables us to use a good heuristic:
		at any point of one of the breadth first-search, we know the current distance from the source can only increase.
		We can stop the search, if this distance divided by the maximal weight (equal to the number of vertices) is larger than the current minimum sparsity value.
		This implementation has \textbf{very good practical performances}: on a regular desktop computer the mean time to compute the sparsest cut 
		of a graph of size $30$ is $0.2$ ms while the brute force algorithm needs $6000$ ms.
	
%
%
%

\section{Metamotifs}\label{sec::metamotifs}

		From a base of motifs, we can generate a new one, by concatenation of elements of the base.
		The new motifs are called metamotifs. It is useful, if the new elements added to the base can be used to remove other elements of the base
		so that some good properties are enforced.
		
		For instance, one can remove the elements of degree $2$ (if they cannot be concatenated together), while not increasing the degree of motifs in the base.
		Every motif of degree $2$ is concatenated in every possible way to the other motifs and deleted. From our example $\{\mot{I},\mot{Y}\}$, we obtain a base $\{\mot{Y}_0,\mot{Y}_1,\mot{Y}_2,\mot{Y}_3\}$ where the $\mot{Y}_i$ are of degree $3$ and have $i$ vertices of $V$ labeled by $\bar{a}$ and the others by $a$.
		If we now generate all molecular maps of size $n$ based on $\{\mot{Y}_0,\mot{Y}_1,\mot{Y}_2,\mot{Y}_3\}$ it is easy to convert them into molecular maps based on $\setofmotifs$. The converted maps are of size exactly $ \frac{5}{2}n$ since there are $\frac{3}{2}$ \mot{I} for each \mot{Y}. 
		
		\begin{figure}
		\begin{center}
		  \begin{tikzpicture}[scale=0.74]
    \node[centre]   (C1) at (-1, 0) {$\mot{Y}$};
    \node[label] (A1) at ( 0,-1) {$\mathbf{\overline{a}}$};
    \node[label] (A2) at ( 0, 0) {$\mathbf{\overline{a}}$};
    \node[labelsat] (A3) at ( 0, 1) {$\mathbf{\overline{a}}$};
    \draw (C1) -- (A1);
    \draw (C1) -- (A2);
    \draw (C1) -- (A3);

    \node[centre]   (C2) at (2,1) {$\mot{I}$};
    \node[labelsat] (B1) at (1,1) {$\mathbf{a}$};
    \node[label] (B2) at (3,1) {$\mathbf{a}$};
    \draw (C2) -- (B1);
    \draw (C2) -- (B2);

         \draw (A3) -- (B1);
  \end{tikzpicture}
  \hspace{2cm}
  \begin{tikzpicture}[scale=0.74]
		  
    \node[centre]   (C1) at (-1, 0) {$\mot{Y}$};
    \node[labelsat] (A1) at ( 0,-1) {$\mathbf{\overline{a}}$};
    \node[labelsat] (A2) at ( 0, 0) {$\mathbf{\overline{a}}$};
    \node[labelsat] (A3) at ( 0, 1) {$\mathbf{\overline{a}}$};
    \draw (C1) -- (A1);
    \draw (C1) -- (A2);
    \draw (C1) -- (A3);

    \node[centre]   (C2) at (2,1) {$\mot{I}$};
    \node[labelsat] (B1) at (1,1) {$\mathbf{a}$};
    \node[label] (B2) at (3,1) {$\mathbf{a}$};
    \draw (C2) -- (B1);
    \draw (C2) -- (B2);

    \node[centre]   (C3) at (2,0) {$\mot{I}$};
    \node[labelsat] (D1) at (1,0) {$\mathbf{a}$};
    \node[label] (D2) at (3,0) {$\mathbf{a}$};
    \draw (C3) -- (D1);
    \draw (C3) -- (D2);

    \node[centre]   (C4) at (2,-1) {$\mot{I}$};
    \node[labelsat] (E1) at (1,-1) {$\mathbf{a}$};
    \node[label] (E2) at (3,-1) {$\mathbf{a}$};
    \draw (C4) -- (E1);
    \draw (C4) -- (E2);

        \draw (A3) -- (B1);
        \draw (A2) -- (D1);
        \draw (A1) -- (E1);
    
  \end{tikzpicture}
		\caption{Representation of the two metamotifs $\mot{Y}_1$ and $\mot{Y}_3$, built from $\mot{Y}$ and $\mot{I}$ }
		\end{center}
\end{figure}
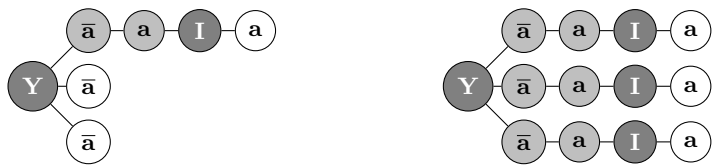
		
		Note that the isomorphism test is done on the generated map of motifs seen as made of the motifs of the first base,
		otherwise we could not detect some isomorphic copies.
		
		The choice of a new base can also be interesting if it decreases its size or the size of the alphabet.
		It is a way to encode constraints on some specific base understood by the user. 
		For instance the base \mot{X} ($a,a,a,a$), \mot{V} ($\bar{a},\bar{a},b$) and \mot{I} ($\bar{b},\bar{b}$) can be turned into the base 
		\mot{X1} ($a,a,a,a$), \mot{X2} ($\bar{a},\bar{a},\bar{a},\bar{a}$) because with \mot{I} we can only connect two \mot{V}.
		It is now easy to see that we are generating the $4$-regular planar bipartite maps. In that particular case, the efficiency of our algorithm is not improved since the generated paths are the same. 
	
\section{Results\label{sec::results}}

The code and the exhaustive results of our approach can be found at the following address
\texttt{http://kekule.prism.uvsq.fr}. For several sets of motifs, one can find the set of generated maps and
their indices. We stopped all computations at $300$ seconds an put a -- in the tables when the algorithm has not finished.
All times are given in second, a.f. stands for almost foldable.

{\small
\begin{table}
\caption{Number of backbones and generation time for \mot{J}($a,b$), \mot{V1}($\bar{a},\bar{a},b$), \mot{V2}($a,\bar{b},\bar{b}$)}
\centerline{
\begin{tabular}[h]{|c||r|r||r|r||r|r||r|r|}
\hline
 Size & \multicolumn{2}{c||}{\textbf{Tree} }  &\multicolumn{2}{c||}{\textbf{A.f. tree} } &\multicolumn{2}{c||}{\textbf{Path} }& \multicolumn{2}{c|}{\textbf{A.f. path}} \\
\hline
    & Backbones & Time & Backbones & Time & Backbones & Time & Backbones & Time \\ 
\hline
   9 & $5.70~10^5$ &   0.09 & $3.85~10^5$ &  0.05 & $4.92~10^4$ &  0.01 & $9.87~10^3$ &  0.01 \\
  12 & $1.16~10^8$ &  14.28 & $5.55~10^7$ &  7.98 & $1.77~10^6$ &  0.28 & $2.46~10^5$ &  0.08 \\
  15 &          -- &     -- &          -- &    -- & $7.26~10^7$ & 10.88 & $6.17~10^6$ &  1.74 \\
  18 &          -- &     -- &          -- &    -- &          -- &    -- & $1.56~10^8$ & 45.84 \\
\hline
\end{tabular}
}
\label{table::backbone}
\end{table}
}

In Tab.~\ref{table::backbone}, we give the time to compute \emph{the backbones} and the number of backbones generated (we also count isomorphic copies which are generated). The time to compute cycles is not given since they are computed from paths, 
the difference is seen in the number of folded maps and the time to generate them.
{\small
\begin{table}
\caption{Number of maps and time to generate them and their indices for \mot{J}($a,b$), \mot{V1}($\bar{a},\bar{a},b$), \mot{V2}($a,\bar{b},\bar{b}$)}
\centerline{
\begin{tabular}[h]{|c||r|r|r||r|r|r||r|r|r|}
\hline
 Size & \multicolumn{3}{c||}{\textbf{A.f. tree} }  &\multicolumn{3}{c||}{\textbf{A.f. path} } & \multicolumn{3}{c|}{\textbf{A.f. cycle}} \\
\hline
    & A.f. backb. & Maps & Time & A.f. backb. & Maps & Time & A.f. backb. & Maps & Time \\ 
\hline
   9 &    $3.85~10^5$ &  236 & 0.32 &  $9.87~10^3$ &  236 &  0.03 & $8.06~10^3$ &    148 &   0.01 \\
  12 &    $5.55~10^7$ &  4476 & 53.99 &  $2.46~10^5$ & 4463 &  0.71 & $2.03~10^5$ &   1931 &   0.32 \\
  15 &        -- & $> 98100$ &   -- &   $6.17~10^6$ & 97112 &  28.40 & $5.13~10^6$ &   29164 &  8.81 \\
  18 &        -- & -- &   -- & $1.56~10^8$ & 2307686 &  -- & $1.30~10^8$ &   501503 &  184.48 \\
\hline
\end{tabular}
}
\label{table::maps}
\end{table}
}
In Tab.~\ref{table::maps}, we give the time to generate \emph{all unique maps} and their indices.
Remark that the number of unique maps generated by trees, paths or cycles are different, since only the generation from trees is exhaustive. 
However, most of the maps with the largest minimum sparsity are generated with paths or cycles as backbones.
%

\section{Chemical validation}\label{sec::validation}

Using the set of motifs $\{\mot{X},\mot{I}\}$, if we take for each size of maps the
ones with the lowest cut indices, we find the molecules obtained by
Warmuth and Liu (Solvent effects in thermodynamically controlled multicomponent nanocage syntheses) 
in real-life experiments. An example of a  molecular map built on $\{\mot{X},\mot{I}\}$ is given in Fig.~\ref{fig:warmuth} (in 3 dimension for easier reading).
The white elements are \mot{X} and the red \mot{I}. Its chemical realization by Warmuth and Liu is also given in the same figure.

\begin{figure}[h]
 \centering
 \includegraphics[scale=0.22]{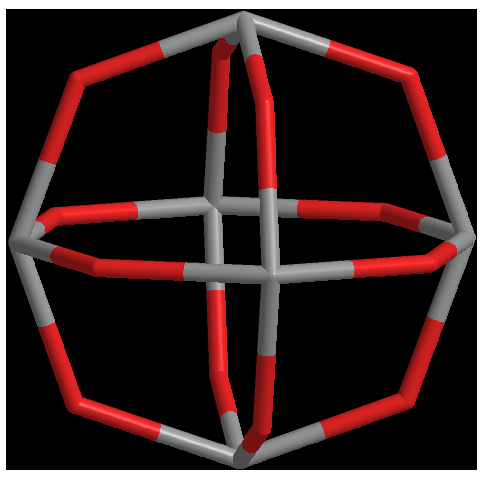}
  \hspace{2cm}
\includegraphics[scale=0.17]{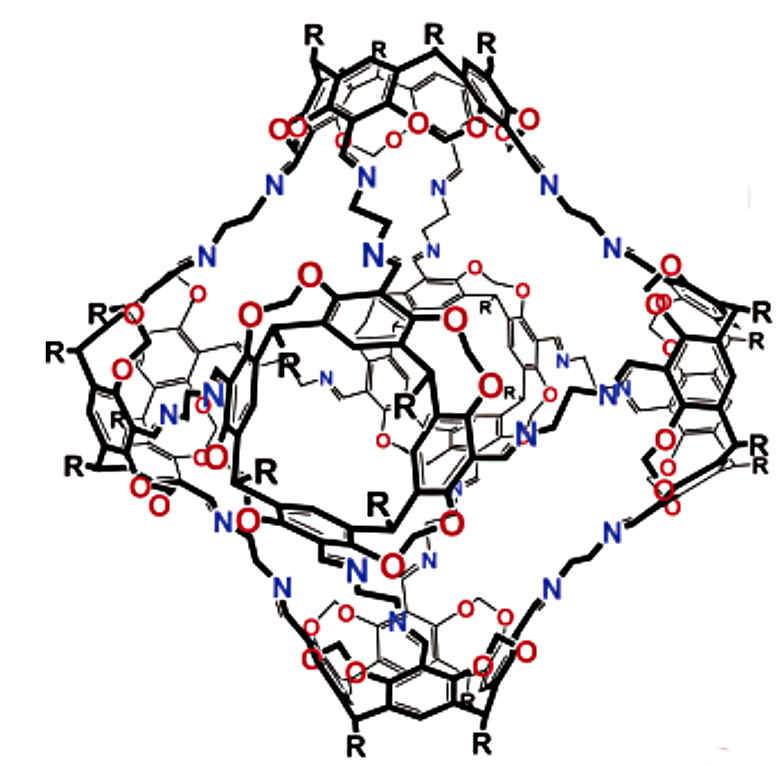}
 \caption{A cage obtained by Warmuth with $6$ \mot{X} and $12$ \mot{I} }
 \label{fig:warmuth}
\end{figure}

From all maps of size $8$ based on \mot{Y} ($a,a,a$), \mot{V1} ($\bar{a},b,b$) and \mot{V2} ($\bar{a},\bar{b},\bar{b}$),
we have selected the map of Fig.~\ref{fig:oliviermol} because it has good indices. This has led to the conception of a real molecule which can be represented by this molecular map. 
It is given in Fig.~\ref{fig:oliviermol}, the blue parts being the \mot{Y}, the black parts the \mot{V1} and the green parts the \mot{V2}.

\begin{figure}[h]
 \centering
 \begin{tikzpicture}[scale=0.70]
   \node[centre]   (Y1) at (0, 0) {$\mot{Y}$};
   \node[centre]   (V11) at (1.5,0) {$\mot{V1}$};
   \node[centre]   (V21) at (3,0) {$\mot{V2}$};
   \node[centre]   (Y2) at (4.5,0) {$\mot{Y}$};
   \node[centre]   (V12) at (1, 1.5) {$\mot{V1}$};
   \node[centre]   (V13) at (1, -1.5) {$\mot{V1}$};
   \node[centre]   (V22) at (3.5, 1.5) {$\mot{V2}$};
   \node[centre]   (V23) at (3.5, -1.5) {$\mot{V2}$};
    \draw (Y1) -- (V11);
    \draw (V11) -- (V21);
    \draw (V21) -- (Y2);
    \draw (Y1) -- (V12);
    \draw (Y1) -- (V13);
    \draw (Y2) -- (V22);
    \draw (Y2) -- (V23);
    \draw (V22) -- (V12);
    \draw (V23) -- (V13);
    \draw (V12) -- (V21);
    \draw (V11) -- (V23);
    \draw (V22) to[out= 120 ,in= 60] (-0.5,1.75) to[out= 240 ,in= 180]  (V13);
  \end{tikzpicture}
  \hspace{2cm}
\includegraphics[scale=0.4]{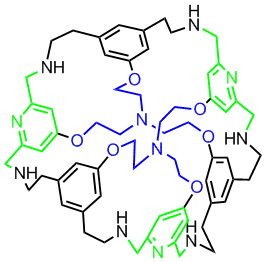}
 \caption{A cage based on $\{\mot{I},\mot{V1},\mot{V2}\}$}
 \label{fig:oliviermol}
\end{figure}

\bibliographystyle{splncs}
\bibliography{biblio.bib}

\end{document}